\documentclass[12pt]{article}
\usepackage{amsmath,amssymb,cite}
\usepackage[french]{babel}
\newtheorem{theorem}{Theorem}

\newtheorem{corollary}[theorem]{Corollary}

\newtheorem{definition}[theorem]{Definition}

\newtheorem{lemma}[theorem]{Lemma}

\newtheorem{proposition}[theorem]{Proposition}
\newtheorem{remark}[theorem]{Remark}

\newcommand{\nix}[1]{}
\DeclareMathOperator{\ord}{ord} 

\newenvironment{proof}[1][Proof]{\textbf{#1.} }{\ \rule{0.5em}{0.5em}}
\usepackage{setspace}
\author{Kenza Guenda
\footnote{Faculty of Mathematics USTHB, University of Sciences
    and Technology of
    Algiers, B.P 32 El Alia, Bab Ezzouar, Algiers, Algeria}} 
\title{New $MDS$ Self-Dual Codes over Large Finite Fields}

\begin{document}
\maketitle
\begin{abstract}
Nous avons construit des codes  $MDS$ qui sont auto-duaux au sens Euclidiens et Hermitiens sur de grands  corps finis. Nos codes sont d\'eriv\'es des codes duadiques cycliques et n\'egacycliques.  
\end{abstract}
\begin{center}
\textbf{Abstract}
\end{center}

We construct $MDS$ Euclidean and Hermitian self-dual codes over large
finite fields of odd and even characteristics. Our codes arise from
cyclic and negacyclic duadic codes.

\section{Introduction}
Let $q$ be a prime power, $\mathbb{F}_q$ a finite field with $q$
elements. An $[n,k]$ linear code $C$ over $\mathbb{F}_q$ is a
$k$-dimensional subspace of $\mathbb{F}_q^n$. A linear code of $\mathbb{F}_q^n$ is said to be constacyclic if it is an
ideal of the quotient ring $R_n=\frac{\mathbb{F}_q[x]}{x^n-a}$. When $a =1$
the code is called cyclic and when $a=-1$ the codes is called negacyclic. For $\mathsf{x} \in C$,
the Hamming weight $wt(\mathsf{x})$ is the number of non zeros
coordinates in $wt(\mathsf{x})$. The minimum distance $d$ of $C$ is
defined as $d=\min\{wt(\mathsf{x}):\, 0 \neq \mathsf{x} \in C\}$.  If
the parameters of the code $C$ verify $n-k+1=d$, then the code is said  to be
maximum distance separable ($MDS$). The minimum distance of a code
is related to its capacity of correctability. For that the $MDS$ codes are optimum in this
sense. Furthermore, the $MDS$ codes find application in algebraic
geometry~\cite{lachaud}. They are also related to geometric objects
called $n$-arcs and to combinatorial objects called orthogonal
arrays~\cite[Ch. 11]{macwilliams}.

 The Euclidean dual code $C^{\bot}$ of the code $C$ is defined as 
$C^{\bot}=\{\mathsf{x} \in \mathbb{F}_q^n: \sum_{i=1}^{n}x_iy_i=0\,
,\forall \mathsf{y} \in C\}$.
If $q=p^2$ the Hermitian dual code $C^{\bot h}$ of  $C$
 is defined as $C^{\bot h}=\{\mathsf{x} \in
\mathbb{F}_q^n: \sum_{i=1}^{n}x_iy_i^p=0\,, \forall \mathsf{y} \in C\}$. 
A code is called Euclidean self-dual or Hermitian self-dual if it
satisfies $C=C^{\bot}$ or $C=C^{\bot h}$. For $q\equiv 1 \mod 4$ a
self-dual code over $\mathbb{F}_q$ exists if and only if $n$ is even,
and for $q\equiv 3 \mod 4$ a self-dual code over $\mathbb{F}_q$ exists
if and only if $n\equiv 0 \mod 4$~\cite[Ch. 19]{macwilliams}. The
linear codes which are close to the Gilbert-Varshamov bound are good
and interesting for practical uses. It turns out that the self-dual codes 
codes satisfies a modified Gilbert-Varshamov-bound~\cite{pierce}. This
makes this family of codes a very attractable family.  

This paper is devoted to the construction of $MDS$ Euclidean and
Hermitian self-dual codes from cyclic duadic codes and negacyclic
duadic codes. Our results, also can be seen as construction of $MDS$ self-dual codes over large fields. This subject
is at the heart of many recent research
papers~\cite{betsumiya,gulliver09,kotsireas}. Whereas, all these constructions care about the computational complexity, our does not.
For that, we reach optimum parameters.
 
Recently, Gulliver et
al.~\cite{gulliver08} constructed $MDS$ self-dual codes of length $q$ over
$\mathbb{F}_q$ from extended Reed-Solomon codes, whenever $q=2^m$. It
turns out, that these codes are extended duadic codes. For $q$ odd the
construction of~\cite{gulliver08} is impossible~\cite{macwilliams}
p. 598. Our construction is more general for fields of odd or even
characteristics, and allows us to construct $MDS$ Euclidean and
Hermitian self-dual codes from the cyclic duadic codes with various lengths. Blackford~\cite{blackford} studied the
negacyclic codes over finite fields by using the multipliers. He gave conditions on the
existence of Euclidean self-dual codes. We generalize his work to the Hermitian
case. We give necessary and sufficient conditions on the existence of
Hermitian self-dual codes from the negacyclic codes. Hence, by using
our previous results, the decomposition of the polynomial $x^n+1$ and
the results of Blackford we construct new
$MDS$ Euclidean and Hermitian self-dual codes from the negacyclic
duadic codes. Furthermore, we give
conditions on the existence of Euclidean self-dual codes which are extended
negacylic.        

The paper is organized as follows: In Section 2 we construct
$MDS$ self-dual codes (Euclidean and Hermitian) from cyclic
duadic codes. First we give cyclic $MDS$ codes over
$\mathbb{F}_q$,
when $n$ divides $q-1$ and $n$ divides $q^2+1$. Furthermore, by using a result from~\cite{guenda} on the existence of
the $\mu_{-q}$ splitting we give 
extended duadic codes which are new $MDS$ Euclidean or Hermitian
self-dual  codes. 

In Section 3 we generalize the work of~\cite{blackford} to the
Hermitian cases. We give necessary and sufficient condition on the
existence of negacylic Hermitian self-dual codes. We construct
negacylic $MDS$ self-dual codes for both the Euclidean and the Hermitian cases.

Several examples have been given from of our results paper. Some of
them reach the best known bounds or even exceed the one given by Kim and
Lee~\cite[Table 1]{kim} or the ones in~\cite{betsumiya,gulliver09,kotsireas}. 
We close by giving a table of $MDS$ self-dual codes of length 18 over prime fields.
\section{Duadic $MDS$ Self-Dual Codes}
Assmus and Mattson~\cite{assmus} have shown that every cyclic code of
 prime length $t$ over $\mathbb{F}_{p^i}$ is $MDS$ for all $i$, except
 a finite number of primes $p$. For cyclic codes of composite length,
Pedersen and Dahl~\cite{pederson} proved that when $n$ divides
 $q=p^h$, there is no-trivial $MDS$ cyclic code
 over $\mathbb{F}_{p^h}$ if and only if $h=1$. In this case any cyclic
 code is $MDS$ and with generator polynomial $g(x)=(x-1)^{p-k}$. For the previous reasons
we consider only the case $(n,q)=1$ and $q$ a prime power. The integer $n$ can be prime or composite and we propose the following Lemma.
\begin{lemma}
\label{number1}
Let $q$ be a prime power. Then if
$n$ divides $q-1$, the polynomial $g_j(x)=\prod_{i=j}^{n-k+j
-1}(x-\alpha^i)$ generate generate an $MDS$ code.
\end{lemma}
\begin{proof}
If $n$ divides $q-1$ i.e., $\textrm{ord }_n(q)=1$, then each
cyclotomic class modulo $n$ contains exactly one element. For a fixed
$k>0$, and $\alpha$ an $n^{th}$ root of unity the  polynomial
$g_j(x)=\prod_{i=j}^{n-k+j
-1}(x-\alpha^i)$ generate an $MDS$ cyclic code,
and that because $g$ has $n-k$ consecutive roots, by the $BCH$ bound
the minimum distance $d$ is such that $d \geq n-k+1$, and then we have
the equality by the Singleton bound.
\end{proof}
\subsection{Euclidean Self-dual $MDS$ Codes over $\mathbb{F}_q$}
This section shows that there exists $MDS$ Euclidean self-dual codes over
$\mathbb{F}_q$ and which arise from cyclic duadic codes.

The following Lemma is useful for the next.
\begin{lemma}\label{lem:Th.4.4.9}
(\cite[Proposition 4.7,Theorem 4.4.9]{lilibeth,huffman03}),
Let $C$ be an $[n,k]$ cyclic code over $\mathbb{F}_q$ with defining
set $T \subset Z_n=\{0,1,\ldots, n-1\}$. Then the following holds:

$(i)$ The Euclidean dual $C^{\bot}$ is also cyclic and has defining
set $Z_n \setminus (-1)T$.

$(ii)$ The Hermitian dual $C^{\bot h}$ is also cyclic and has defining
set $Z_n \setminus (-q)T.$
\end{lemma}

When we consider an odd integer $n$ which divides $q-1$, hence $q$ is
a residue quadratic modulo $n$,( denoted by $q=\square \mod n$). Then
from~\cite[Theorem 1]{smid87}, there exists a duadic code of length
$n$. Now we will construct some of these duadic codes. Consider the following cyclic code $D_1$ with defining set
$T_1=\{1,2,\ldots, \frac{(n-1)}{2}\}$. By Lemma~\ref{number1}, the code $D_1$ is
an $[n,\frac{(n+1)}{2},\frac{(n+1)}{2}]$, $MDS$ code over
$\mathbb{F}_q$, by Lemma~\ref{lem:Th.4.4.9} its dual $C_1=D_1^{\bot}$
is also cyclic with defining set $T_1\cup\{0\}$. The code $C_1$ is
self-orthogonal as  $T_1 \subset T_1\cup\{0\}$ and is with dimension
$\frac{n-1}{2}$ and with minimum distance $\frac{n-1}{2}$, hence also
$MDS$.

This gives that the code $C_1$ is an even-like duadic code whose
splitting is given by $\mu_{-1}$ due to the following Lemma.
\begin{lemma}\label{lem:Th.6.4.1}(~\cite[Theorem 6.4.1]{huffman03})
Let $C$ be any $[n,\frac{n-1}{2}]$ cyclic code over $\mathbb{F}_q$,
  with $q$ a prime power. Then $C$ is self-orthogonal if and only if
  $C$ is an even like duadic code whose splitting is given by $\mu_{-1}$.
\end{lemma}
This gives us a pair of duadic codes $D_1= C_1^{\bot}$ and
$D_2=C_2^{\bot}$ and a pair of even like duadic code
$C_2=\mu_{-1}(C_1)$.
Hence the following result.
\begin{lemma}
\label{lem:main1}
Let $n$ be an odd integer which divides $q-1$, then there exists a pair of $MDS$
codes $D_1$, $D_2$ with parameters
$[n,\frac{(n+1)}{2},\frac{(n+1)}{2}]$, which are duadic codes with the
splitting given by $\mu_{-1}$.
\end{lemma}

Since $n$ is odd, we want to extend the codes $D_i$ for $1\leq i \leq
2$ in such a way the extended code is self-dual. This is possible
provided the hypothesis of the following Lemma are satisfied. 
\begin{lemma}\label{lem:Th.6.4.12}(~\cite[Theorem 6.4.12]{huffman03})
Let $D_1$ and $D_2$ be a pair of odd-like duadic codes of length $n$
over $\mathbb{F}_q$. Assume that 
\begin{equation} \label{eq:sol}
1+ \gamma ^2 n=0
\end{equation}
 has a solution in
$\mathbb{F}_q$. Then if $\mu_{-1}$ gives the splitting from $D_1$ and
$D_2$, then $\widetilde{D_1}$ and $\widetilde{D_2}$ are self-dual. Here $\widetilde{D_i}=\{\widetilde{\mathsf{c}}\, | \, \mathsf{c}
\in D_i\}$ for $1 \leq i \leq 2$ and $\widetilde{\mathsf{c}}=c_0\ldots
c_n c_{\infty}$ with $c_{\infty}=-\gamma \sum_{i=0}^{n-1}c_i$.
\end{lemma}
 In general it is not always possible to
 find a solution to the equation~(\ref{eq:sol})
 in $\mathbb{F}_q$. Furthermore, extending an $MDS$ code does not give
 always an $MDS$ code. But under some conditions this can be possible, as
 proved by Hill~\cite{hill}. 
For $n=q-1$ we have $\gamma=1$ is a solution
 of~(\ref{eq:sol}). Moreover, if the code is a Reed-Solomon code, then
by a result of Macwilliams and
 Sloane~\cite[Theorem.10.3.1]{macwilliams} the extended code is also
 $MDS$. In the landmark textbook~\cite{huffman03} the solution of the 
 equation~(\ref{eq:sol}) is discussed, when $n$ is an odd prime
 number.
The following Lemma generalize their results to $n=p^m$. 
\begin{lemma}
\label{lem:order}
Let $q=r^t$, with $r$ an odd prime, $t$ an odd integer and $n=p^m$
such that $n$ divides $q-1$. Then there is a solution to the 
equation~(\ref{eq:sol}) in $\mathbb{F}_{q}$, whenever one of the following holds.
\begin{enumerate}
\item $r\equiv 3 \mod 4$, $p\equiv 3 \mod 4$ and $m$ odd;
\item $r\equiv 1\mod 4$ and $p \equiv 1 \text{ or } 3 \mod 4$; 
\end{enumerate}  
\end{lemma}
\begin{proof}
As mentioned before if $n$ divides $q-1$, then $q=\square \mod
p$. This gives that $q=\square \mod r$.  
Hence if $ p \equiv 3 \mod 4$ and  $r\equiv 3 \mod 4$, there is a solution $\gamma$ to
the equation~ $1+\gamma^2p=0$~\cite[Lemma 6.6.17]{huffman03}. If $m$ is
odd, hence $\gamma^m$ is a solution to the equation~(\ref{eq:sol}).
Now, assume $q \equiv 1 \mod 4$ and $p\equiv 1 \text{ or }3 \mod 4$. The equation~
$1+\gamma^2p=0$~\cite[Lemma 6.6.17]{huffman03} has a solution in
$\mathbb{F}_q$. As for the previous case, if $m$ is odd $\gamma^m$ is
a solution to the equation~(\ref{eq:sol}). Now, assume that $m$ is even, since for
such $p$ and $q$ there is a solution to
$1+\gamma^2p=0$ in $\mathbb{F}_q $~\cite[Lemma 6.6.17]{huffman03}. This gives $(\gamma^m)^2=\frac{1}{p^m}$. But, since $r \equiv 1 \mod
4$, then $-1$ is a quadratic residue in $\mathbb{F}_r \subset \mathbb{F}_q $~\cite[Lemma 6.2.4]{huffman03}. Then, there exists
an $a \in \mathbb{F}_q $, such that $a^2=-1$. Hence $a\gamma^m$ is a
solution of the equation~(\ref{eq:sol}) in $\mathbb{F}_q$. 
\end{proof}

In the following result we give Euclidean self-dual codes which are
$MDS$.
\begin{theorem}
\label{th:main2}
Let $q=r^t$ be a prime power (even or odd), $n$ an odd divisor of
$q-1$. Then there exists a pair of $D_1,D_2$ of $MDS$ odd-like duadic
codes, with splitting $\mu_{-1}$ and where the even-like duadic codes
are $MDS$ self-orthogonal and $T_1=\{1,\ldots,\frac{n-1}{2}\}$.
Furthermore, the following holds:

$(i)$ If $q=2^t$, with $t$ odd and $n=p$ an odd prime, then the extended codes $\widetilde{D_i}$ are
$[n+1,\frac{n+1}{2},\frac{n+3}{2}]$ $MDS$ and Euclidean self-dual codes.

$(ii)$ If $q=r^t$, with $t$ even and $n$ odd and divides $r-1$, then the
extended codes $\widetilde{D_i}$ for $1\leq i\leq 2$ are $[n+1,\frac{n+1}{2},\frac{n+3}{2}]$ $MDS$ Euclidean self-dual codes.

$(iii)$ If $q=r^t$, with $r\equiv 3 \mod 4$, $t$ odd and $n=p^m$, with
$p$ a prime such that $p \equiv 3\mod 4$ and $m$ is odd, then the extended codes
$\widetilde{D_i}$ are $[n+1,\frac{n+1}{2},\frac{n+3}{2}]$ $MDS$ and Euclidean
self-dual codes.

$(iv)$ If $q=r^t$, with $t$ odd, $p$ a prime such that $r \equiv p \equiv 1 \mod 4$ and $n=p^m$, then the extended codes
$\widetilde{D_i}$ are $[n+1,\frac{n+1}{2},\frac{n+3}{2}]$ $MDS$ and Euclidean self-dual codes.
\end{theorem}  
\begin{proof}
Lemma~\ref{lem:main1} gives a pair $D_1,D_2$ of $MDS$ odd-like duadic
codes, with splitting $\mu_{-1}$ and where the even-like duadic codes
are $MDS$ self-orthogonal and $T_1=\{1,\ldots,\frac{n-1}{2}\}$.
If $q=2^t$, $t$ odd and $n=p$ an odd prime which divides $q-1$, hence
$q=\square \mod n$. From~\cite[Lemma 6.6.17]{huffman03}, there is a
solution to the equation~(\ref{eq:sol}) in~$\mathbb{F}_q$. Hence from Lemma~\ref{lem:Th.6.4.12}, the extended codes $\widetilde{D_i}$ are
self-dual. If $t$ is even and $n$ is an odd integer which divides
$r-1$, from~\cite{huffman03} p. 227 there is a solution
of the equation~(\ref{eq:sol}) in $\mathbb{F}_{r^2} \subset
\mathbb{F}_q$,  since the coefficients are in
$\mathbb{F}_{r}$.   
Further, if we assume $r\equiv3 \mod 4$, $t$ odd and $n=p^m$  with $m$
odd and such that $p \equiv 3\mod 4$,  by Lemma~\ref{lem:order}, there
is a solution to the equation~(\ref{eq:sol}). Hence from Lemma~\ref{lem:Th.6.4.12} the
extended codes $\widetilde{D_i}$ are self-dual. Similarly if we
assume $r\equiv 1\mod 4$, $t$ odd and $n=p^m$ such that $p \equiv 1
\text{ or }3\mod 4$, we have a solution to the equation~(\ref{eq:sol}) by Lemma~\ref{lem:order}. Hence from Lemma~\ref{lem:Th.6.4.12} the
extended codes $\widetilde{D_i}$ are self-dual.  
Now we prove that $\widetilde{D_i}$ are $MDS$. Let $\mathsf{c}$ be a
codeword of $D_i$ of weight $\frac{n+1}{2}$, the minimum weight of
$\widetilde{D_i}$ is increasing by 1 provided 
\begin{equation}
\label{eq:mds}
-\gamma \mathsf{c}(1)=-\gamma \sum_{i=0}^{n-1}c_i= c_{\infty} \neq 0.
\end{equation}
But $\gamma \neq 0$, hence to have~(\ref{eq:mds}) it suffices to verify that $\mathsf{c}(1) \neq 0$.
Since $c(x)=a(x)g(x)$ for some $a(x)\mod (x^n-1)$ and $g(x)=
\prod_{i=1}^{\frac{n-1}{2}}(x-\alpha^i)$. $g(1) \neq 0$, also $a(1) \neq 0$
otherwise, $a(x)$ is a multiple of $(x-1)g(x)$. Hence by the $BCH$
bound the weight is $\geq 1+\frac{n+1}{2}$, by the singleton bound we
get the equality.
\end{proof}

\begin{table}
\begin{center}
\begin{tabular}{|c|c|c|c|}
\hline $n$&$q$&$n$&$q$ \\
\hline
4&7,13,19,31,43,49,79,97,$11^2$,$13^2$,$17^2$,$31^2$&6&17,$9^2$,$11^2$,$31^2$ \\
8&8,29,43,71,$13^2$,$2^9$&10&19,37,73,109,$19^2$\\
12&23,67,89&14&$53^2$\\
16&$31^2$&18&$103^2$\\24&$2^{11}$&30&$59^2$\\
32&32,$5^3$&74&293,$2^9$\\
84&167&90&$2^{11}$\\
\hline
\end{tabular}
\end{center}
\caption{Euclidean Self-dual $MDS$ Codes over $\mathbb{F}_q$ obtained by Theorem~\ref{th:main2}}
\end{table}

\subsection{Hermitian Self-Dual $MDS$ Codes}
Let $q$ be a power of an odd prime $r$. In this part we will construct $MDS$ self-dual codes over
$\mathbb{F}_{q^2}$ of length $n+1$, with $n|q^2+1$.

First remark that, when $n$ divides $q^2+1$, then the multiplicative
order of $q^2$ modulo $n$ is equal to 2. This implies that all the
cyclotomic classes $C(i)$ modulo $n$ are reversible with 
cardinality  1 or 2, that is because
$|C(i)|$ divides $\ord _nq^2$. It follows that, if $n$ is odd then
$C(i)=\{i,-i\}$ for  any $i\neq 0$.    
 If we consider the cyclic code generated by the polynomial 
 $$g_s(x)=\prod_{i= \frac{n-1}{2}-s}^{i= \frac{n-1}{2}+s+1}(x-\alpha^i) \text{ with } 0 \leq s \leq
\frac{n-1}{2},$$ 
it is an $[n, n-2s-2, 2s+3]$ $MDS$ cyclic code. The polynomial $g_s(x)$ has $2s+2$ consecutive roots
$$\alpha^{\frac{n-1}{2}-s},\alpha^{\frac{n-1}{2}-s+1}, \ldots,
\alpha^{\frac{n-1}{2}+1},\ldots , \alpha^{\frac{n-1}{2}+s+1}.$$
This gives a cyclic $MDS$ code with odd dimension $k=n-2s-2$.

Now, consider $n=p^m$ such that $n$ divides $q^2+1$ and $p^m \equiv 1
\mod 4$, for $s=\frac{n-1}{4}-1$, the polynomial $g_s$ generate a
cyclic $MDS$ code $D_1$ of parameter
$[n,\frac{n+1}{2},\frac{n+1}{2}]$.
 Lemma~\ref{lem:Th.4.4.9} gives that Hermitian dual of $D_1$ is with defining set $Z_n\setminus
(-qT)$. Since $\ord_nq^2$ is even (equal to 2), then the multiplier $\mu_{-q}$
gives a splitting~\cite[Proposition 13]{guenda09a}. Hence the code
$D_1$ is one of the odd-like duadic codes and then $D_1^{\bot h}=C_1$ is the even like duadic
with defining set $T\cup \{0\}$. Hence $C_i \subset C_i^{\bot
  h}=D_i$. As for the Euclidian case, using the usual extension of an
orthogonal code does not
give always a self-dual code. If we consider in 
$\mathbb{F}_{q^2}$ the equation 
\begin{equation}
\label{eq:dual}
1+ \gamma ^{q+1}n=0,
\end{equation} it has always a solution in $\mathbb{F}_{q^2}$ if we
assume $n \in \mathbb{F}_{r}$ as mentioned
in~\cite{lilibeth}.
For $1 \leq i \leq 2$, the extended codes are $\widetilde{D_i}=\{\widetilde{\mathsf{c}}\, | \, \mathsf{c}
\in D_i\}$, with $\widetilde{\mathsf{c}}=c_0\ldots
c_n c_{\infty}$, $c_{\infty}=-\gamma
\sum_{i=0}^{n-1}c_i$ and $\gamma$ is solution of the equation~(\ref{eq:dual}).

Since in this case the splitting is given by $\mu_{-q}$, the codes
$\widetilde{D_i}$ are Hermitian self-dual~\cite[Proposition
  4.8]{lilibeth}. By a similar argument as in Theorem~\ref{th:main2},
the extended codes are also $MDS$, since the codes $D_i$ are $MDS$.
This prove the following Theorem.
\begin{theorem}
\label{th:main3}
Let $q=r^t$ be a prime power, $n=p^m \in \mathbb{F}_r$ a divisor of $q^2+1$, where $p^m
\equiv 1 \mod 4$. There exists Hermitian self-dual codes over $\mathbb{F}_{q^2}$ which are
$MDS$ and extended duadic codes with the splitting given
$\mu_{-q}$ and with parameters $[n+1,\frac{n+1}{2},\frac{n+3}{2}]$. 
\end{theorem}

\begin{table}
\begin{center}
\begin{tabular}{|c|c|c|c|}
\hline $n$&$q$&$n$&$q$ \\
\hline
6&3,7,13,17,23,37,43,47,53,63,67,73,83&
14&31,47,73,83
\\18&13&30&17\\
38&31&42&73\\54&23,83&
62&11\\138&
37&182&19
\\234&89&422&29\\
\hline
\end{tabular}
\end{center}
\caption{Hermitian Self-dual $MDS$ Codes over $\mathbb{F}_{q^2}$ obtained by Theorem~\ref{th:main3}}
\end{table}
\section{Negacyclic Duadic Codes}
It was proved in ~\cite{hai} that if $n$ is odd, then the negacyclic
codes are equivalent to cyclic codes, for that we consider only
negacyclic codes with even length.

Now in order to use it latter we review the factorization of the
polynomial $x^n+1$ over $\mathbb{F}_q[x]$. This can be found
in~\cite{nuh,krishna,jensen}. We also assume $(n,q)=1$, so that the
polynomial $x^n+1$ does not have multiple roots.

The roots of $x^n+1$ are $\delta,\delta \xi,\ldots, \delta \xi^{n-1}$,
 where $\xi$ is a primitive $n$th root of unity and
 $\delta^n=-1$. Hence $\xi= \delta^2$, $\delta$ is a primitive $2n$th
 root of unity. Hence $\delta$ lies in an extension field
 $\mathbb{F}_{q^s}$, with $s$ equal to the multiplicative order of $q$
 modulo $2n$. Let $\omega$ be a primitive element of
 $\mathbb{F}_{q^s}$, hence  we can take $\delta = \omega^t$ and $\xi = \omega
^{2t}$, with $t=\frac{q^s-1}{2n}$. Then the following holds.
$$x^n+1= \prod_{i=0}^{n-1}(x-\delta \xi ^i)= \prod_{i=0}^{n-1}
(x-\omega ^{t(1+2i)})= \prod_{i=0}^{n-1}
(x-\delta^{(1+2i)}).$$
To each irreducible factor of $x^n+1$ corresponds a cyclotomic class modulo $2n$.
$\delta ^{2i+1}$ and $\delta ^{2j+1}$ are said to be conjugate if they
   are roots of a same irreducible factor of $x^n+1$. 

If we denote by $O_{2n}$ the set of odd integers from $1$ to
 $2n-1$. The defining set of negacyclic code $C$ of length $n$ is
 the set $T=\{i \in O_{2n}:\, \delta^i \text{ is a root of }C\}$. It
 will be the union of $q$-cyclotomic classes modulo $2n$. The
 dimension of the negaclic code with defining set $T$ is $n-|T|$. 
Nuh et al. ~~\cite{nuh} gave a negacylic 
 $BCH$ bound given. That is if $T$ has $d-1$
 consecutive odd integers, then the minimum distance is at least $d$.     

\begin{lemma}\label{lem:th2.blackford}
(\cite[Theorem 2]{blackford}
If $C$ is a negacyclic code with defining set $T$, then $C^{\bot}$
(the Euclidian dual of $C$) is a negacyclic code with defining set 
$$T^{\bot}=\{i \in O_{2n}\, :\, -i (\mod 2n) \notin T\}$$
\end{lemma}
Let $s \in \{1, \ldots, 2n-1\}$ such that $(s,2n)=1$, a multiplier of
$R_n$ is the map: 
\begin{equation}
\begin{split}
\mu _s : R_n  &  \longrightarrow \mathbb{F}_q^n  \\
a(x) &\longmapsto \mu_s(a(x)) (\mod x^n+1),
\end{split}
\end{equation}
$\mu _s$ is an  automorphism of $R_n$. If $C$ is an ideal of
$R_n$ with defining set $T$ , then $\mu _s(C)$ is an ideal of
$R_n$ with defining set $\{i \in O_{2n} \,|\, si \in
T \}$. $\mu _s$ induces the following map
\begin{equation}
\begin{split}
\mu _s' : O_{2n}  &  \longrightarrow  O_{2n}  \\
i &\longmapsto \mu_s'(i)=si (\mod 2n),
\end{split}
\end{equation}

The multiplier $\mu _{2n-1}=\mu _{-1}$ has the effect to replace $x$
by $x^{-1}$, since $x^{2n}=1$ in $R_n$.

\begin{lemma}\label{lem:th.3.blackford}
(\cite[Theorem 3]{blackford})
If $N=2^an'$ for some odd integer $n'$, then self-dual negacyclic codes
over $\mathbb{F}_q$ of length $N$ exist if and only if
$$q \neq -1 (\mod 2^{a+1}).$$
If $a=1$, then self-dual negacyclic codes over $\mathbb{F}_{q}$ of
length $N$ exist if and only if
$$q\equiv 1 \mod 4.$$
\end{lemma}
 As a Corollary of Lemma~\ref{lem:th.3.blackford} the negacyclic code of length $q+1$ and defining set $T= \{i \text{
  odd}:\, 1\leq i \leq q\}$ is an Euclidean self-dual $MDS$ code over
$\mathbb{F}_q$ as proved in~\cite{blackford}. The following results is
  more general than the ones given in~\cite{blackford} . 
\begin{theorem}
\label{Cor:Euc11}
Let $n=2n'$ for some odd integer $n'$, $q$ an odd prime power such
that $q\equiv 1 \mod 4$, $q+1=2n''$, with $n'|n''$ and $n''$ odd. Then there exists
$MDS$ negacyclic Euclidean self-dual code of
parameters $[n,n/2,n/2+1]$ having defining set
$$T=\{\frac{q+1}{2}+i:\, -(n'-1)\leq i \text{ even } \leq (n'-1)\}.$$
\end{theorem}
\begin{proof}
Consider a negacyclic code $C$ with such length $n$ over $\mathbb{F}_q$. Assume $\delta^{2i'+1}$ is a root of $C$,
hence
$(\delta^{2i'+1})^{q+1}=
\delta^{2i'(q+1)}\delta^{q+1}=\delta^{2jn}\delta^{q+1}=\delta^{q+1}$. 
Then for an odd $i\in O_{2n}$ the conjugate of $\delta^i$ is $\delta^{iq}= \delta^{q+1-i}$. 
Hence we have $C(i)=\{i,q+1-i\}$. It is clear that for $i \in O_{2n}$
we have $|C(i)|\leq 2$. And $i=q+1-i \mod 2n \iff i= \frac{q+1}{2}+kn.$
Hence for $i$ even, such that $ 1\leq i \leq
(n'-1)$ we have
$|C(\frac{q+1}{2}+i)|=|\{\frac{q+1}{2}+i,\frac{q+1}{2}-i\}|=2$ and
for $i=0$ $|C(\frac{q+1}{2})|=1$. 
Now, consider a negacyclic code with the following defining set: 
$$T=\cup_{i=0}^{n'-1}C(\frac{q+1}{2}+i)=\{\frac{q+1}{2}+i:\, -(n'-1)\leq i \text{ even } \leq
(n'-1)\}.$$ Assume there exists  two differents integers $i$ and $j$ such that
 $0 \leq i \leq n'-1$, $0 \leq j \leq n'-1$ and
$C(\frac{q+1}{2}+i)=C(\frac{q+1}{2}+j)$. Hence
$\frac{q+1}{2}+i=\frac{q+1}{2}+j+2kn \iff i-j=2kn$. That is $i-j$ is a
multiple of $2n$. But we have $i-j \leq n$, which is impossible.
Furthermore, from~Lemma~\ref{lem:th.3.blackford} we have
$C(i)\neq C(-i) \mod 2n$. If we assume the existence of two different
$i'$, $j'$ in $T$ such that $C(i')=C(-j')$, hence there exists $i$
and $j$ such that $i'=\frac{q+1}{2}+i$ and $j'=\frac{q+1}{2}+j$. But,
$C(i')=C(-j')\iff \frac{q+1}{2}+i=2kn -\frac{q+1}{2}-j \iff -(q+1+ 2k'n
)= i+j=n(-\frac{q+1}{n}+2k),$ this gives that $n$ divides $i+j$,
which is impossible since $-(n'-1)\leq i,j \leq
(n'-1)$. This implies, that $-T \cap T= \emptyset$ and the redundancy of the code is equal to $n'$, hence
the code is self-dual. The
code is $MDS$, since there is $n'$ successive roots and hence by the
$BCH$ bound the minimum distance is at least $n'+1$, hence by the
Singleton bound we have equality.  
\end{proof}

\begin{table}
\begin{center}
\begin{tabular}{|c|c|c|c|c|c|}
\hline $n$&$q$&$n$&$q$&$n$&$q$ \\
\hline
6&5,17,29,53,197&10&9,29,49,$13^2$&
14&13\\
18&17,53,89,101,197& 
22&109,197&26&25,181,233\\
30&29,89,149&34&101,$13^2$&
38&37,113\\42&41,293,461&50&49,149,&54&53,269\\
\hline
\end{tabular}
\end{center}
\caption{Euclidean Self-dual $MDS$ Codes over $\mathbb{F}_q$ obtained by Theorem~\ref{Cor:Euc11}}
\end{table}

\begin{theorem}
\label{th:Euc12}
Let $n=2^an'$ for some odd integer $n'$, $q$ an odd prime power such
that $q\equiv 1 \mod 2^{a+1}n''$, $n'|n''$ and $n''$odd. Then there exists
$MDS$ negacyclic Euclidean self-dual code of
parameters $[n,n/2,n/2+1]$ having defining set $$T=\{i \text{ odd}:\, 1 \leq i \leq n-1\}.$$
\end{theorem}
\begin{proof}
In this case we have $\xi \in \mathbb{F}_q$, hence $\xi^q=\xi$. We will
show that the conjugate of $\delta^{2i+1}=\delta \xi ^i$ is exactly it
self. This means that each cyclotomic class contains only one
element. Namely $$(\delta \xi ^i)^q=\delta^q \xi
^i=\delta\delta^{q-1} \xi=\delta(\delta^{2n})^\frac{q-1}{2n} \xi ^i= \delta \xi ^i .$$
Now we consider the negacyclic code with defining set $T=\{i
\text{ odd}:\, 1 \leq i \leq n-1\}$, by Lemma~\ref{lem:th.3.blackford} we
obtain $C(i)\neq C(-i)$. Furthermore, for different $i$ and $j$ in $T$
we cannot have $C(i)=C(-j)$. Because, if i is the case we will have
$2nk-i=j$, since in each class there is only one element. Hence
$i+j=2nk$, which is impossible, because $i\leq n-1$ and $j\leq
n-1$. Which implies $-T \cap T=\emptyset$ and
$|T|=\frac{n}{2}$. Then from Lemma~\ref{lem:th2.blackford}, we obtain $T^{\bot}=T$. Hence the code
is self-dual. By the $BCH$ bound the minimum distance is $\frac{n}{2}+1$  
\end{proof}

\begin{table}
\begin{center}
\begin{tabular}{|c|c|c|c|c|c|}
\hline $n$&$q$&$n$&$q$&$n$&$q$ \\
\hline
6&13,25,37,49,61,73&10&41,61,81&12&49,73,97\\
14&29,113&18&37,73,109&20&41,81,121\\
24&49,193&26&53,157&28&169,281,337\\
30&61,121,181&34&409&36&73,433\\
\hline
\end{tabular}
\end{center}
\caption{Euclidean Self-dual $MDS$ Codes over $\mathbb{F}_{q}$ obtained by Theorem~\ref{th:Euc12}}
\end{table} 

\begin{lemma}
\label{lem:herm.dual}
Let $C$ be a negacyclic code of length $n$ over $\mathbb{F}_{q^2}$
with defining set $T$. Then the Hermitian dual is a negacyclic code
with defining set $$T^{\bot h}= O_{2n} \setminus (-iq) T\}.$$
\end{lemma} 
\begin{proof}
Let $\overline{C}=\{(a_0^q,\ldots,a_{n-1}^q)\, : \,
(a_0,\ldots,a_{n-1}) \in C \}$. By an analogous argument
as in ~\cite[Proposition 3.1]{lilibeth} one can show that
$\overline{C}=\mu_{q}(C)$. This gives that the code $\overline{C}$ is
a negacyclic code with defining set $T_{\overline{C}}=qT=\{iq \,:\, i
\in T\}$. By noticing that $C^{\bot h}=\overline{C}^{\bot }$, we
get that $$T^{\bot }_{\overline{C}}=\{i\in O_{2n}\, :\, -i (\mod 2n)
\notin qT\}.$$ Since $\mu_q$ is an automorphism on $R_n$, hence induces
a permutation acting on the elements of $O_{2n}$. Thus we have:
\begin{equation}
\label{eq2:equiv}
-i (\mod 2n) \notin qT  \iff -qi \mod 2n \notin q^2T.
\end{equation}
But over $\mathbb{F}_{q^2}$ all the cyclotomic classes are stable by
multiplication by $q^2$, hence the equation~(\ref{eq2:equiv}) is
equivalent to $-qi \mod 2n \notin T$. Then, 
$$T^{\bot h}=\{i \in O_{2n}\, :\, -iq (\mod 2n) \notin T\}= O_{2n}
\setminus (-q) T\}.$$
\end{proof} 
\begin{proposition}
\label{prop:her.dual}
If $N=2^an'$ for some odd integer $n'$, there exists a Hermitian self-dual code
over $\mathbb{F}_{q^2}$ of length $N$ if and only if 
\begin{equation}
\label{eq:odd}
q \neq -1 \mod 2^{a+1}.
\end{equation}
\end{proposition}
\begin{proof} From Lemma~\ref{lem:herm.dual}, the code $C$ is Hermitian self-dual if
and only if we have $T=O_{2n} \setminus (-iqT)$, hence $C$ is Hermtian
self-dual if its defining set $T$ satisfies
the following 
\begin{equation}
2N -iq \notin T \iff i \in T.
\end{equation}
Then, if there exists an odd $i \in O_{2N}$, such that
$C_{q^2}(i)=C_{q^2}(-qi) \mod 2N$, the code $C$ is not self-dual. If a
such $i$ exists, then there is an integer $m$ such that $-iq \equiv q^{2m} i (\mod
2N)$. Hence,
$2^{a+1}n'k=(q^{2m-1}+1)qi$ i.e., $2^{a+1}n' | (q^{2m-1}+1)qi$. Since
$n'$ is odd we can choose $i$ such that $n' \equiv i \mod 2N)$. 
We need only check that $2^{a+1} |(q^{2m-1}+1)q$. Since $q$ is odd
hence
$2^{a+1} |  q^{2m-1}+1$. Thus it is sufficient only to check that $q\equiv -1
\mod 2^{a+1}$.
 \end{proof}

For $a=1$, the equation~(\ref{eq:odd}) becomes $q \equiv 1 \mod 4$,
hence the following Corollary. 
\begin{corollary}
If $N=2n'$, for some integer $n'$, then Hermitian self-dual negacyclic codes
over $\mathbb{F}_{q^2}$ of length $n$ exist if and only if $$q \equiv1
\mod 4.$$
\end{corollary}
\begin{theorem}
\label{th.main4}
Let $n=2^an'$, $a>1$ and $q \equiv 1 \mod 2^an''$, such that $n' |n''$         
and $n''$ odd. Then there exists an $MDS$ negacyclic codes which is
Hermitian self-dual with defining set
$$T= \{i \text{ odd}: 1 \leq i \leq n-1 \}.$$
\end{theorem}
\begin{proof}
If $q\equiv 1 \mod 2^an''$, then $q \neq -1 \mod 2^{a+1}$. Because if
it is the case, then $q=-1+k2^{a+1}$ and $q=1+2^an''k'$, by  summing
the two quantities of $q$ and dividing by 2 both sides, we have
$q=2^{a-1}(n''k'+2k)$. This implies that $q$ is even, since $a>1$, which is impossible. Hence by Proposition~\ref{prop:her.dual}
  we have $C_{q^2}(-qi)\neq \{i\}$, since we proved that $q\neq
  -1+k2^{a+1}$. For these parameters we have $\xi \in
  \mathbb{F}_{q^2}$. Then by a similar argument as Theorem~\ref{th:Euc12} we
  have that $C_{q^2}(i)=\{i\}$. Now, we prove that for $i,j \in T$ we
  cannot have $C_{q^2}(-qi)=\{j\}$. Assume it is the case,
  then we
  will have $2kn-qi=j$. The last equality is equivalent to
  $2^{a+1}n'k-2^an''k'=i+j \iff 2^a(2k-\frac{n''k'}{n'})=i+j$. Hence
  $2^an'$ divides $i+j$. But $i$ and $j$ odd gives
  $i+j=2(1+2k'')$. This gives a contradiction since we assumed $a>1$. Hence by Lemma~\ref{lem:herm.dual} we obtain $T^{\bot
    h}=T$. Hence the code is Hermitian self-dual. By the $BCH$ bound
  the minimum distance is $n/2+1$.   
\end{proof}

\begin{table}
\begin{center}
\begin{tabular}{|c|c|c|c|c|c|}
\hline $n$&$q$&$n$&$q$&$n$&$q$ \\
\hline
12&13,$5^2$,37,$7^2$,97&20&41,61,81,101,181&
24&$5^2$,$7^2$,73,97,121\\
28&29,113,197&36&37,73,109&
40&41,$9^2$,$11^2$\\
42&43,127&48&$7^2$,97&60&61,181\\
44&89,353&48&97,193,241,281,337&52&53,157,313,\\
\hline
\end{tabular}
\end{center}
\caption{Hermitian Self-dual $MDS$ Codes over $\mathbb{F}_{q^2}$ obtained by Theorem~\ref{th.main4}}
\end{table} 
A generalization of the splitting of $n$ to the negacyclic codes whose
introduced in~\cite{blackford}.

A $q$ splitting of $n$ is a multiplier $\mu_s$ of $n$ that induce a
partition of $O_{2n}$, such that
\begin{enumerate}
\item $O_{2n} = A_1 \cup A_2 \cup X$
\item $S_1$, $S_2$ and $X$ are unions of $q$ cyclotomic classes.
 \item $\mu_s'(S_i)=S_{i+1 (\mod 2)}$ and $\mu_s'(X)=X$. 
\end{enumerate}
A $q$ splitting is of type $I$, if $X= \emptyset$. A $q$
splitting is of type $II$ if $X=\{\frac{n}{2},\frac{3n}{2}\}.$
\begin{definition}
A negacyclic code $C$ of length $n$ over $\mathbb{F}_{q}$ is 
duadic if there exists a such splitting and the defining set
is one of the subset $S_i$ or  $S_i\cup X$. If the splitting is of
type II, then there exists polynomials $A_i(x)$ such that
$x^n+1=A_1(x)A_2(x)(x^2+1)$ and $\mu_s(A_i(x))=A_{i+1}(x)$.
\end{definition} 
\begin{remark}
An Euclidean respectively Hermitian self-dual negacyclic code is
duadic with multiplier $\mu_{-1}$ respectively $\mu_{-q}$ and comes from type $I$ splitting.
\end{remark}

In the next we consider negacyclic code with length $n=2p^t$, with $p$
an odd prime.
\begin{lemma}\label{lem:existence.split}(\cite[Theorem 8]{blackford})
If $p$, $q$ are distinct odd primes, $q\equiv 3 \mod 4$ and $r=\ord_{2p^t}$ is the order of $q$ modulo $2p^t$, then we have the following which holds.
\begin{enumerate}
\item
 There exists a $q$ splitting of $n=2p^t$ of type $II$. 
\item $\mu_{-1}$ gives a splitting of $n$ of type $II$ if and only if
  $r\neq 2 \mod 4$.
\end{enumerate} 
\end{lemma}
\begin{remark}
\label{rem:1}
 We have $r=\ord_{2p^t}q= lcm(\ord_{2}q,\ord_{p^t}q)=
\ord_{p^t}q$, since $q$ is odd. Let $z$ be the largest integer such that $p^z | (q^t-1)$, with
$t$ order of $q$ modulo $p$. Hence if $z=1$,
we have $\ord_{p^t}q=p^{t-1}\ord_pq$~\cite[Lemma 3.5.4]{guenda}.
Hence if $\ord_pq$ is odd or $\ord_pq\equiv 0 \mod 4$, then $r \neq 2
\mod 4$. Hence from Lemma~\ref{lem:existence.split} the multiplier $\mu_{-1}$
gives a splitting of $n$ of type $II$.
\end{remark}
\begin{lemma}
\label{lem:quadra} Let $p$ and $q$ be odd prime number, hence we have the following.
\begin{enumerate}
\item If $p \equiv 1 \mod 4$ and $(\frac{q}{p})=-1$, hence $\ord_p q \equiv 0 \mod 4$.
\item If $(\frac{q}{p})=1$ and $p\equiv 3 \mod 4$, hence $\ord_p q$ is odd.
\end{enumerate} 
\end{lemma}  
\begin{proof}
If we assume that $q$ is not a quadratic residue modulo $p$. Hence
from~\cite[Lemma 6.2.2]{huffman03} $\ord_p q$ is not a divisor of $\frac{p-1}{2}$.
Then from Fermat's Theorem $\ord_p q=p-1$, hence $\ord_p q\equiv 0
\mod 4$, since $p \equiv 1 \mod 4$. 

If $q=\square \mod p$, hence from~\cite[Lemma
  6.2.2]{huffman03} $\ord_pq$ is a divisor of $\frac{p-1}{2}$. Since,
$p\equiv 3 \mod 4$, then $\frac{p-1}{2}$ odd which implies $\ord_p q$
is also odd.     
\end{proof}

Assume, that the following equation 
\begin{equation}
\label{eq:equiv}
2+ \gamma ^{2} n =0
\end{equation}
has a solution in $\mathbb{F}_{q}$. 
If $\mathsf{a}=(a_0,\ldots,a_{n-1}) \in \mathbb{F}_{q}^{n}$, define
$$\widetilde{\mathsf{a}}=(a_0,\ldots,a_{n-1},a_{\infty},a_*) \in
\mathbb{F}_{q}^{n+1},$$
where $$a_{\infty}=\gamma \sum_{i=0}^{\frac{n-1}{2}} (-1)^ia_{2i},
\,\,\,a_*=\gamma \sum_{i=0}^{\frac{n-1}{2}} (-1)^ia_{2i+1}.$$
The set $\widetilde{C}=\{\widetilde{\mathsf{a}}=(a_0,\ldots,a_{n-1},a_{\infty},a_*) \in
\mathbb{F}_{q}^{n+1}: (a_0,\ldots,a_{n-1}) \in C\}$ is
a linear code of $\mathbb{F}_{q}$. 
\begin{lemma}\label{Thm:12.black}
(\cite[Theorem 12]{blackford})

Let $q$ be a prime power and $\gamma$ is a solution of the 
equation~(\ref{eq:equiv}) in $\mathbb{F}_{q}$, and suppose that $D_1$ and $D_2$ are odd-like
negacyclic duadic codes of length $n=2p^t$, with multiplier $\mu_{-1}$ of type $II$. Then
$\widetilde{D_i}$ for $i=1,2$ are self-dual. 
\end{lemma}
\begin{lemma}
\label{lem:sol}
Let $q,p$ be odd prime such that $q\equiv p \equiv 3 \mod 4$,
$n=2p^t$, with $t$ odd. Hence the equation~(\ref{eq:equiv}) has a solution in $\mathbb{F}_{q}$ 
\end{lemma}
\begin{proof}
There is a solution for the equation $2+2p\gamma^2=0$ in $\mathbb{F}_q$ if and only
if there is a solution of $1+p\gamma^2=0$ in $\mathbb{F}_q$. If we
assume $p \equiv 3 \mod 4$, the last equation has a solution $\gamma
\in  \mathbb{F}_q $ from~\cite[Lemma 6.6.17]{huffman03}. If $t$ is odd
 $\gamma^t$ is a solution of the equation~(\ref{eq:equiv}). 
\end{proof} 
\begin{theorem}
\label{Th:main.nega}
 Let $p,q,$ be two odd primes such that $q=\square \mod p$, $q\equiv p
 \equiv 3 \mod 4$ and $z=1$. Then there exists negacyclic duadic codes of length
 $n=2p^t$, $t$ odd, with
 splitting of type $II$ given by $\mu_{-1}$, and such that
 $\widetilde{D_i}$ are self-dual for $i=1$ and $2$.  
\end{theorem} 
\begin{proof}
If we have such $p$ and $q$ from Lemma~\ref{lem:quadra} the $\ord_{2p^t}q$
is odd. Hence from Remark~\ref{rem:1} $\mu_{-1}$ gives a splitting of
$n$ of type $II$. Furthermore, from Lemma~\ref{lem:sol} the
equation~(\ref{eq:equiv}) has a solution in $\mathbb{F}_q$. Hence from
Lemma~\ref{Thm:12.black} The codes $D_i$ are extended to self-dual
codes $\widetilde{D_i}$, for $i=1$ and $2$. 
\end{proof}

If we assume $ q\equiv 3 \mod 4$, $p \equiv 1 \mod 4$, $z=1$ and $q$ not
residue quadratic modulo $p$, hence from Lemma~\ref{lem:quadra} the
$\ord_{2p^t}q\equiv 0 \mod 4$. Hence from Remark~\ref{rem:1}
$\mu_{-1}$ gives a splitting of $n$ of type $II$. Hence there exist
duadic negacyclic codes. Unfortunately in this case we cannot extend
the codes in order to get self-dual codes as we done in
Theorem~\ref{Th:main.nega}. That is because simply the equation~(\ref{eq:equiv}) 
 has no solution in $\mathbb{F}_q$. Namely, a solution will implies
 that $(\frac{-p^t}{q})=1$. But, we have 
 $(\frac{-p^t}{q})= (\frac{-1}{q})(\frac{p^t}{q}).$
Since $q \equiv 3 \mod 4$, hence $-1$ is not a quadratic residue modulo
$q$. Then  $(\frac{-p^t}{q})= -(\frac{p^t}{q})$. Furthermore, from the
law of quadratic reciprocity, we have
$(\frac{p}{q})(\frac{q}{p})=(-1)^{\frac{p-1}{2}\frac{q-1}{2}}=-1$.
 As $(\frac{q}{p})=-1$, this implies $(\frac{p}{q})=1$. Hence
 $(\frac{p^t}{q})=1$.

Gulliver and Harada~\cite{gulliver09} proved the existence of $MDS$
self-dual codes of length $18$  over  $\mathbb{F}_p$, whenever $17\leq
97$. But when $101\leq p \leq 300$ they gave quasi-twisted self-dual
$[18,9,9]_p$ from unimoduallar lattices~\cite[Table 3]{gulliver09}.
 In the following table we give some examples of
$MDS$ self-dual codes of length $18$, for $p \geq 101$.  
\begin{table}
\begin{center}
\begin{tabular}{|c|c|}
\hline $p$ & Arguments \\
\hline
109& Theorem~\ref{th:Euc12} \\
137& Theorem~\ref{th:main2}\\
181& Theorem~\ref{th:Euc12} \\
197& Theorem~\ref{Cor:Euc11} \\
233& Theorem~\ref{Cor:Euc11} \\
269& Theorem~\ref{Cor:Euc11} \\
\hline
\end{tabular}
\end{center}
\caption{$MDS$ self-dual $[18,9,10]$ over $\mathbb{F}_p$}
\end{table}
\section*{Acknowledgment}
The author would like to thank Professor Thierry P. Berger for help
and encouragement in carrying out the research in this paper.

{\small

\end{document}